\documentclass[twoside]{IEEEtran}

\ifCLASSINFOpdf
\else
\fi

\usepackage{color}
\usepackage{amsthm}
\usepackage{textcomp}

\usepackage{amsmath}
\usepackage{amssymb}
\usepackage{graphicx}
\usepackage{tikz}
\usepackage{multicol}
\usepackage{pifont}
\usepackage{enumerate}

\usepackage{textcomp}

\usepackage{amsmath}
\usepackage{amssymb}
\usepackage{graphicx}
\usepackage{tikz}
\usepackage{multicol}
\usepackage{pifont}
\usepackage{enumerate}

\newtheorem{myproposition}{\it Proposition}

\newtheorem{mydefinition}{\it Definition}
\newtheorem{myexample}{\it Example}

\begin{document}
\title{Strange Attractor in Density Evolution}

\author{\IEEEauthorblockN{Sinan Kahraman}
\IEEEauthorblockA{Department of Electrical-Electronics Engineering\\
Bilkent University,
Ankara, TR-06800, Turkey\\
Email: sinank@ieee.org}
}

\author{Sinan~Kahraman,~\IEEEmembership{Member,~IEEE}
\thanks{
This work was supported in part by T\"UB\.ITAK,
Turkey, under 1929B011500065.
S.~Kahraman was with Bilkent University, Ankara, TR-06800, Turkey
(e-mail: sinank@ieee.org).
}}

\onecolumn

\markboth{DRAFT}%
{DRAFT}

\maketitle

\begin{abstract}
The strange attractor represents a complex pattern of behavior in dynamic systems.  
This paper introduces a strange attractor for synthetic channels in polar coding as a result of a geometric property of density evolution that is a polar code construction technique.

First, we define a subset of synthetic channels that are universally less reliable than the original channel.
Here, the cardinality of the attractor set is $(n+2)$-th Fibonacci number for the block length $N=2^n$. 
This can be seen as a significantly large number for very long codes.
On the other hand, strange attractor can provide new achievable rates for the finite block lengths.  

Secondly, it is known that polar codes can be constructed with sub-linear complexity by the use of partial orderings.  
In this study, we additionally define $1+\log_2 (\log_2 N)$ universal operators to reduce the complexity. 
Then, these universal operators can be applied on the attractor set to increase the number of synthetic channels that are universally less reliable than the natural channel. 
\end{abstract}


\IEEEpeerreviewmaketitle

\section{Introduction} 

Polar coding is the first and only coding technique to provably achieve the channel capacity for binary discrete memoryless channels using quasi-linear complexity encoding, decoding and code construction methods defined in detail \cite{arikan_channel_2009}. This technique received great interest due to this important advantage. It has been discussed in 3GPP standardization works and accepted to be used in 5G technology. From an industry point of view, this result demonstrates that polar coding can be considered for use in different technologies where long code lengths are preferred to achieve higher reliabilities for a fixed code rate (i.e., size of information set divided by code length.) This paper aims to design very long polar codes. 

In polar coding, conventional code construction is defined as determination of the order of reliabilities of all synthetic channels. Using Monte-Carlo simulation is a way for the code construction described in \cite{arikan_channel_2009}. Later, since polar coding is a channel specific technique, the polar code construction has been studied as a research direction in the literature. This is mainly due to the fact that once the polar code is designed for communication systems, it is necessary to make this design specific to the channel. For this reason, various code construction methods based on calculating the reliability of the synthetic channels are discussed by the density evolution \cite{mori_density_eval}, upgrading and downgrading
\cite{tal_how_2013}, and Gaussian approximation \cite{trifonov_gaussian}. A comperative study in \cite{harish} investigates the performance of these polar code constructions.


Recently in \cite{schurch_partial} and \cite{barded_partial}, a partial order for synthetic channels is defined as an universal (channel independent) property of the channel polarization. This feature has been considered in \cite{mondelli_sublinear} to reduce the complexity of polar code design based on the considered calculations. As a result, it has been shown in \cite{mondelli_sublinear} that code design for polar codes can be done with very low complexity such as a sub-linear complexity.

The structure of the polar codes in \cite{arikan_channel_2009} with the block length $N$ is introduced by $G=F^{\otimes n}$ matrix. It is defined by the $n^{th}$ Kronecker power of $2\times 2$ kernel matrix $F$. The encoding task is expressed as $x=uG$, in modulo-2 arithmetic.
Polar coding in \cite{arikan_channel_2009} has low complexity encoder and decoder. Here, FFT-like structures require $\mathcal{O}(N\log N)$ complexity. The input vector $u$ with $N$ length contains $K$ information components and $N-K$ frozen components for the coding rate $R=\frac{K}{N}$. It is assumed that the frozen locations are known by the receiver. Here, $N-K$ synthetic channels that are the lowest reliable are reserved for frozen components. High reliability is considered as large mutual information, small Bhattacharya parameter and small error probability. 
Code design methods provides the locations of frozen components for polar coding. 
The transition probabilities of the synthetic channels obtained after one-step of polarization are defined as follows:
\begin{eqnarray}
W^- &=&W(y_1,y_2|u_1) \nonumber  \\
       &=&\frac{1}{2}\sum_{u_2=0}^{1}W(y_1|u_1\oplus u_2)W(y_2|u_2),   \\
W^+&=&W(y_1,y_2,u_1|u_2) \nonumber \\
       &=&\frac{1}{2}W(y_1|u_1\oplus u_2)W(y_2|u_2),
\end{eqnarray}
where $y_1$ and $y_2$ are noisy observation of the receiver unit, $u_1$ and $u_2$ are the input of one-step polarization. Here, $W^-$ denotes the polarized bad channel. $W^+$ is the polarized good channel. $W^+$ has higher reliability than $W^-$ and this is represented as $W^-\prec W \prec W^+$, where $W$ is the natural channel. Reliability ordering for $N$ synthetic channels depends on the natural channel. For that reason, code designs are channel specific that are based on Monte-Carlo simulation or density evolution calculation by the use of Gaussian approximation. 
The solution to the problem is sufficient to be done only once. Unfortunately, it is channel specific that is the major issue of the code design. Recent researches on the partial order have focused on this problem. They exploited relative non-channel specific solutions of the synthetic channels. 

\section{Channel Ordering}\label{sec3}

As a channel independent method, reliabilities of the some synthetic channels can be universally comparable by using the {\em channel ordering} that is intensively studied in \cite{schurch_partial}, \cite{barded_partial} and \cite{mondelli_sublinear_arxiv}. 

The following notation was used to define this property. Any synthetic channel such as $(\cdots((W^+)^-)^+\cdots)^-=W^{+-+\dots -}$ obtained by $n$-step polarization is mapped to index in $[0,N)$ using $1$ for $+$ and $0$ for $-$ polarization step. 
E.g., $W^{--++}: W_3$ with (0011) binary index and $W^{+--+}: W_9$ with (1001) binary index. 
\newline
Let $k_i$ be the $i^{th}$ most significant bit of the binary index of $k$. 

\begin{mydefinition}[The first order operator]\label{firstorder}
{\it {\bf Addition}}.

If $k_i=1$ and $k_j=\ell_j$ for all $j$ where $j\neq i$, then $W_\ell\preceq W_k$.\label{def1}
\end{mydefinition}

\begin{mydefinition}[The second order operator]\label{secondorder}
{\it {\bf Left swap}}.

If $k_i,k_{i+t}=10$ and $\ell_i,\ell_{i+t}=01$ and also $k_j=\ell_j$ for all $j$ and $t\geq1$ where $j\neq i$ and $j\neq i+t$, then $W_\ell\prec W_k$.\label{def2}
\end{mydefinition}

Simply, we have the results $W_{(ab0c)}\prec W_{(ab1c)}$ and $W_{(a01b)}\prec W_{(a10b)}$ for more clarity.
It was introduced that this partial order technique is a sub-linear complexity code design method in \cite{mondelli_sublinear} and \cite{mondelli_sublinear_arxiv}.

In this study, we first introduce the new partial orders that are also feasible tools for the efficient polar code design to reduce the complexity of design method in \cite{mondelli_sublinear_arxiv}. For this purpose, we first introduce the following result. Then, we show that the proposed new feature can sort synthetic channels with a smaller difference in reliability that can not be separated by the known partial order with the {\it Definition~\ref{firstorder}} and {\it \ref{secondorder}}. It can be noticed that these orderings are still universal. To introduce the multiple partial order, we define new operators as follows:

\begin{mydefinition}[The new partial order]{\bf Multiple.}

For a given universal partial order $W_{\cdots0\cdots}\prec W_{\cdots1\cdots}$ (the first order operator), it is easy to notice that $W_{\cdots01\cdots}\prec W_{\cdots10\cdots}$ (the second order operator) is also universal partial order. Recursively, $W_{\cdots0110\cdots}\prec W_{\cdots1001\cdots}$ and $W_{\cdots01101001\cdots}\prec W_{\cdots10010110\cdots}$ are universal partial orders.
\end{mydefinition}

This property is a natural extension of the left swap operator in {\it Definition~\ref{secondorder}}. The number of operators that can be given for the block length $N$ is $1 + \log_2 (\log_2 N)$. Here, we provide $5$ operators that are given in Table~\ref{tab_n16} for $N=65536$. This result shows that the multiple partial order relations can provide more than two operators for $N\geq16$.  
Following examples can be given as a result of the new feature. 
\begin{myexample}\label{example3rd}
By using $3^{rd}$ order operation, 
\begin{equation}
W_{(0110)}\prec W_{(1001)}.
\end{equation}
\end{myexample}
\begin{myexample}
By using $4^{rd}$ order operation, 
\begin{equation}
W_{(01101001)}\prec W_{(10010110)}.
\end{equation}
\end{myexample}
They are new partial order definitions that can not be obtained by the  {\it Definitions~\ref{firstorder}} and {\it \ref{secondorder}}. Notice that the resolution of the new operators are higher than the previous definitions.
This is an important property that can be exploited to order antichains.

\begin{table}
\caption{Multiple operators for $N=65536$}
\begin{center}
\begin{tabular}{|c|c|}
\hline
order & operator (less reliable $\prec$ more reliable) \cr \hline
$1^{st}$ & $0\prec 1$  \cr \hline
$2^{nd}$ & $01\prec 10$  \cr \hline
$3^{rd}$ & $0110\prec 1001$  \cr \hline
$4^{th}$ & $01101001\prec 10010110$  \cr \hline
$5^{th}$ & $0110100110010110\prec 1001011001101001$  \cr \hline
\end{tabular}
\end{center}
\label{tab_n16}
\end{table}%
\pagebreak

We provide the proof of multiple partial orders. For this purpose, proof for 1001 and 0110 partial order can be provided as follows for any given channel reliabilities $x=L$,
\begin{eqnarray}
g(x)&=&W^-\rightarrow eqn.(\ref{eqnn})\\
f(x)&=&W^+\rightarrow eqn.(\ref{eqnp})
\end{eqnarray}

Here, $f(x)$ and $g(x)$ are monotonic increasing functions  for all $x>0$. It is clear that $f(x)>g(x)$  for all $x>0$. 

Hence, we have
\begin{itemize}
\item[\it i.] $g(f(x))$ and $f(g(x))$ are increasing functions and $f(x)>g(x)$ for all $x>0$. Then, following result is obtained. 
\begin{eqnarray}g(f(x))<f(g(x))\end{eqnarray}
\item[\it ii.] $f(g(g(f(x))))$ and $g(f(f(g(x))))$ are increasing functions and $g(f(x))<f(g(x))$ for all $x>0$. 
\newline Finally, the following result is obtained. 
\begin{eqnarray}g(f(f(g(x))))<f(g(g(f(x))))\end{eqnarray}
\end{itemize}
Hence, the partial order $0110\prec 1001$ is obtained.

This can be successively applied for higher order partial orders. 
Experimental results are placed in Section~\ref{SecResult}.

\section{Improved Gaussian Approximation}\label{sec4}

The Gaussian approximation for density evolution was first proposed by Chung \textsl{et al.} to analyze low density parity check codes in \cite{Chung}. 
Then, the Gaussian approximation was used by Trifonov in \cite{trifonov_gaussian} as one of the deterministic ways to compute the reliability of synthetic channels. 
We summarized this method as the following way.

We assume that the all-zero codeword is transmitted to the receiver. The log-likelihood ratio (LLR) for a noisy observation $y_i=x_i+n_i$ is defined as $L^i_1(y_i)=\log \frac{W(y_i|0)}{W(y_i|1)}$.    
The probability density function is $f(x)=e^{-x^2/2\sigma^2}$ for additive white Gaussian noise with $N(0,\sigma^2)$ distribution. The expected value of $L^i_1(y_i)$ can be considered as follows:
$$E\left[L^i_1(y_i)\right]={{2}/\sigma^2}.$$

Variance of the LLR is given as: 
$$V\left[L^i_1(y_i)\right]=\frac{4}{\sigma^2}.$$

The update rules for the expectations of inter-level LLRs is given for $i=1,\dots,n/2$ as follows: 
\begin{eqnarray}
E\left[L^{(2i-1)}_{j}\right]&=&\phi^{-1}\left(1-\left(1-\phi\left(E\left[L^i_{j/2}\right]\right)\right)^2\right),\label{eqnn} \\
E\left[L^{(2i)}_{j}\right]&=&2E\left[L^i_{j/2}\right]\label{eqnp}
\end{eqnarray}
where

\begin{equation}
\phi(x)=\Bigg \{
  \begin{tabular}{lr}
  $1-\frac{1}{\sqrt{4 \pi x}} \int_{-\infty}^{\infty} \tanh \frac{u}{2} e^{-\frac{(u-x)^2}{4x}} du$ & $x>0$ \\
  $1$, & $x=0$  
  \end{tabular}.
  \end{equation}
  
The error probability of indices $i\in\{1,\dots,N\}$ is given as follows: 
\begin{equation}\pi_i \approx Q\left(\sqrt{E\left[L^i_{N}\right]/2}\right)=\frac{1}{2}\textrm{erfc}\left(\frac{1}{2}\sqrt{E\left[L^i_{N}\right]}\right)\end{equation}
where \begin{equation}\textrm{erfc}\left(x\right)=\frac{2}{\sqrt{\pi}} \int_{x}^{\infty} e^{-v^2} dv.\end{equation}
An upper bound of the error probability is the sum of error probabilities for the set of information indices.

To simplify the update rule we use an approximation of
\begin{equation}\tanh x \approx \left\{
  \begin{tabular}{cr}
  $1$, & $x>0$ \\
  $0$, & $x=0$ \\ 
  $-1$, & $x<0$
  \end{tabular}\right .\end{equation}
as given in the subsection, and hence, the simplified update rule is provided by using the following definitions:
\begin{eqnarray}
\phi(x) =\textrm{erfc}\left(\frac{\sqrt{x}}{2}\right)
\end{eqnarray}
\begin{eqnarray}
{\phi}^{-1}(x) = 4 \left(\textrm{erfcinv}\left(x\right)\right)^2
\end{eqnarray}
The simplified update rule is given as follows:
\begin{eqnarray}
&&E\left[L^{(2i-1)}_{j}\right]=\nonumber\\
&&4\left(\textrm{erfcinv}\left(1-\left(1-\textrm{erfc}\left(\frac{1}{2}\sqrt{E\left[L^i_{j/2}\right]}\right)\right)^2\right)\right)^2\,\,\,\,\,\,
\end{eqnarray}
\begin{eqnarray}
&&E\left[L^{(2i)}_{j}\right]=2E\left[L^i_{j/2}\right].
\end{eqnarray}

This is a numerically stable update that can be efficiently implemented by only using a lookup table for the function $\textrm{erfc(x)}$ and $\textrm{erfcinv(x)}$.   

\subsection{Simplification of the functions: $\phi(x)$ and $\phi^{-1}(x)$}
First, we consider the following assumption:
$$\tanh x \approx \left\{
  \begin{tabular}{cr}
  $1$, & $x>0$ \\
  $0$, & $x=0$ \\ 
  $-1$, & $x<0$
  \end{tabular}\right .$$
Then, we use the following equations.
\begin{eqnarray}
&&\frac{1}{\sqrt{4 \pi x}} \int_{-\infty}^{\infty} \tanh \frac{u}{2} e^{-\frac{(u-x)^2}{4x}} du \approx \nonumber\\
&&\frac{1}{\sqrt{4 \pi x}} \left( \int_{0}^{\infty} e^{-\frac{(u-x)^2}{4x}} du - \int_{-\infty}^{0} e^{-\frac{(u-x)^2}{4x}} du\right) \nonumber
\end{eqnarray}
We apply the transformation: $\frac{u-x}{2\sqrt{x}}=v$. Then, 
\begin{eqnarray}
&&\frac{1}{\sqrt{4 \pi x}} \left( \int_{0}^{\infty} e^{-\frac{(u-x)^2}{4x}} du - \int_{-\infty}^{0} e^{-\frac{(u-x)^2}{4x}} du\right)=\nonumber\\
&&\frac{1}{\sqrt{\pi}} \left( \int_{-\sqrt{x}/2}^{\infty} e^{-v^2} dv - \int_{-\infty}^{-\sqrt{x}/2} e^{-v^2} dv\right).\nonumber
\end{eqnarray}
Then, we use the definition: $$\textrm{erfc}\left(\frac{\sqrt{x}}{2}\right)=\frac{2}{\sqrt{\pi}} \int_{\sqrt{x}/2}^{\infty} e^{-v^2} dv.$$  
\begin{eqnarray}
&&1-\textrm{erfc}\left(\frac{\sqrt{x}}{2}\right)=\nonumber\\
&&\frac{1}{\sqrt{4 \pi x}} \left( \int_{0}^{\infty} e^{-\frac{(u-x)^2}{4x}} du - \int_{-\infty}^{0} e^{-\frac{(u-x)^2}{4x}} du\right).\nonumber
\end{eqnarray}
Finally, we have the simplified equations as follows:
$$\phi\left(x\right) = \textrm{erfc}\left(\frac{\sqrt{x}}{2}\right),$$

$${\phi}^{-1}\left(x\right) = 4 \left(\textrm{erfcinv}\left(x\right)\right)^2.$$

Here, we can compute the reliability of synthetic channels by the simplified Gaussian approximation update functions. Now, we investigate the behaviors of these update functions by using the geometric properties.
First, $y=2x$ and $y=\phi^{-1}\left(1-\left(1-\phi\left(x\right)\right)^2\right)$ functions are depicted in Fig.~\ref{fig1}. Moreover, the reflections of these curves with respect to the $y=x$ line are also added.

Experimental result for AWGN channel is given in Section~\ref{SecResult} for channel ordering 1001 and 0110.

\begin{figure}[b!]
\centering
\includegraphics[width=0.5\textwidth]{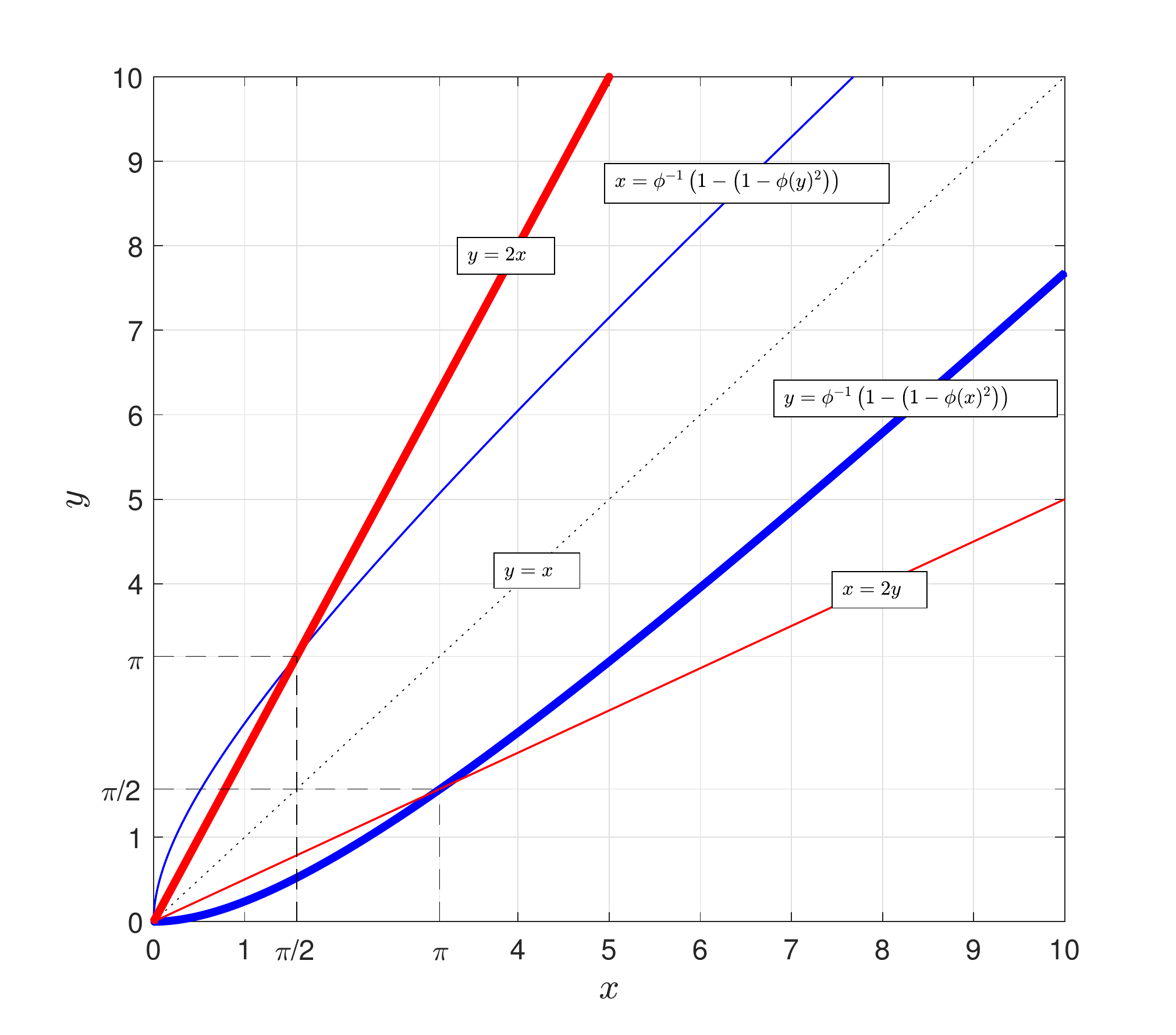}
\caption{Plot of the recursive functions for update rule of Gaussian approximation method. (bold curves: the functions and thin curves: the reflections.)} 
\label{fig1}
\end{figure} 
\newpage
\section{Strange Attractor}\label{sec5}
In this section, we focus on geometric properties of the update rules that are considered in the previous section for Gaussian approximation for polar code constructions. 

Let us define the functions $f_1(x)=x/2$ and  $f_2(x)=\phi^{-1}\left(1-\left(1-\phi\left(x\right)\right)^2\right)$. 
In Fig.~\ref{fig1}, some observations can be noted as the following properties:

\begin{enumerate}[i)]
\item $y=f_1(x)$ and $y=f_2(x)$ intersect at $(x=0,y=0)$.
\item $y=f_1(x)$ and $y=f_2(x)$ intersect at $(x=\pi,y=\pi/2)$.
\item $f_1(x)>f_2(x)$ for $x \in (0,\pi)$.
\end{enumerate}

These observations help us to identify an universal subset of synthetic channels that are less reliable than the natural channel. As a similar approach to analyze in a chaotic systems, we introduce strange attractor to represent the complex pattern of behavior in channel polarization. We provide the following definition for channel polarization scenario.

\begin{mydefinition}
Strange attractor is a geometric property to define a subset of synthetic channels that are universally unreliable than the natural channel $W_k\prec W$. 
\end{mydefinition}

As a brief description, any natural channel for a given LLR can be polarized to a synthetic channels with the index that has no '11' in binary expansion provides a subset of synthetic channels as a strange attractor.

Moreover, it can be noticed that they are universally less reliable than the natural channel with $LLR<\pi/2$. If the block length is long enough, the reliability of these channels converge to 0. We provide the following two examples to make a connection between strange attractor and polar code construction. 

\begin{myexample}
For the original channel with $LLR<\frac{\pi}{2}$, all possible synthetic channels that are labelled by the indices with binary expansion $(k_1,k_2,\dots,k_n)$ with $k_i\neq1$ and $k_{i+1}\neq1$ for any $i=[1,n)$ are less reliable than the original channel.
\end{myexample}
\begin{myexample}\label{example_pi}
For the original channel with $LLR<\pi$, all possible synthetic channels that are labelled by the indices with binary expansion $(k_1,k_2,\dots,k_n)$ with $k_i\neq1$ and $k_{i+1}\neq1$ for any $i=[1,n)$ and $k_1\neq1$ are less reliable than the original channel.
\end{myexample}

\begin{myproposition}\label{th_attractor}
As the block length increases, the LLR values of the synthetic channels that are labelled by the indices with binary expansion $(k_1,k_2,\dots,k_n)$ with $k_i\neq1$ and $k_{i+1}\neq1$ for any $i=[1,n)$ converges to $0$ (unreliable) and the number of these type of synthetic channels is $F_{2+n}$, where $F_i$ is $i^{th}$  Fibonacci number in $\{1,1,2,3,5,8,13,21,34,55,89,144,\dots\}$ for the block length $N=2^n$.
\end{myproposition}

\begin{proof}
The proof for the proposition is presented in two parts. The first part is concerned with the exact number of $n$-length bit strings with $k_i\neq1$ and $k_{i+1}\neq1$ for any $i=[1,n)$. Let $\mathcal{A}^i$ be a set of $i$-length bit strings with $k_i\neq1$ and $k_{i+1}\neq1$ for any $i=[1,n)$. The cardinality of the set $|\mathcal{A}^i|$ can be given as follows:

\begin{enumerate}[i)]
\item $|\mathcal{A}^1|=2$ where $\mathcal{A}^1:\{0,1\}$
\item $|\mathcal{A}^2|=3$ where $\mathcal{A}^2:\{00,01,10\}$
\item $|\mathcal{A}^3|=5$ where $\mathcal{A}^3:\{000,001,010,100,101\}$
\item $|\mathcal{A}^4|=8$ where\\ $\mathcal{A}^4:\{{\bf 0}000,{\bf 0}001,{\bf 0}010,{\bf 0}100,{\bf 0}101,{\bf 10}00,{\bf 10}01,{\bf 10}10\}$
\item $|\mathcal{A}^\ell|=|[{\bf 0}|\mathcal{A}^{\ell-1}],[{\bf 10}|\mathcal{A}^{\ell-2}]|=|\mathcal{A}^{\ell-1}|+|\mathcal{A}^{\ell-2}|$ 
\end{enumerate}
As a result, $|\mathcal{A}^n|=F_{n+2}$ 

where $F_n=\{1,1,2,3,5,8,13,21,34,55,89,144,\dots\}$.

The second part of this proof is about the attractor. As a result of the observations, we could pre-define $F_{n+2}$ synthetic channels named as attractor thanks to the geometrical properties in Fig.~\ref{fig1} we obtained. We observe that LLR values converge to 0 for all possible bit strings with $k_i\neq1$ and $k_{i+1}\neq1$ for any $i=[1,n)$. 
\end{proof}

%

We notice that similar observation is mentioned as tangent bifurcation or saddle node that was reported in \cite{reportedd} to compute decoding thresholds and analyze LDPC codes. We use the similar observation as in \cite{vv} to define a complex pattern of behavior by using the Strange Attractor definition in this study.  

\section{Improved Polar Code Design}

We first consider strange attractor for the finite block length in this section. The number of the attracted synthetic channels are provided for a given polarization steps $n$ in Table~\ref{tab1}. 

\begin{table}[hp]
\caption{Number of this type of synthetic channels}
\begin{center}
\begin{tabular}{|c|c|c|c|c|}
\hline
$n$ & $F_{n-2}$ & Rate1 & PO & Rate2  \cr \hline
6 & 21 & 0.6719         & 17 & 0.4063 \cr \hline
7 & 34 & 0.7344         & 26 & 0.5313 \cr \hline
8 & 55 & 0.7852         & 31 & 0.6641 \cr \hline
9 & 89 & 0.8262         & 38 & 0.7520 \cr \hline
10 & 144 & 0.8594     & 31 & 0.8291 \cr \hline
11 & 233 & 0.8862     & 40 & 0.8667 \cr \hline
12 & 377 & 0.9080     & 31 & 0.9004 \cr \hline
13 & 610 & 0.9255     &      &             \cr \hline
14 & 987 & 0.9398     &      &             \cr \hline
15 & 1597 & 0.9513   &      &             \cr \hline
16 & 2584 & 0.9606   &      &             \cr \hline
\end{tabular}
\end{center}
\label{tab1}
\end{table}%

Rate1 in Table.\ref{tab1} show the achievable rates by only strange attractor. 
Alternatively, the achievable rates are shown in the following figure for $LLR<\pi/2$.
\begin{figure}[hp]
\centering
\includegraphics[width=0.35\textwidth]{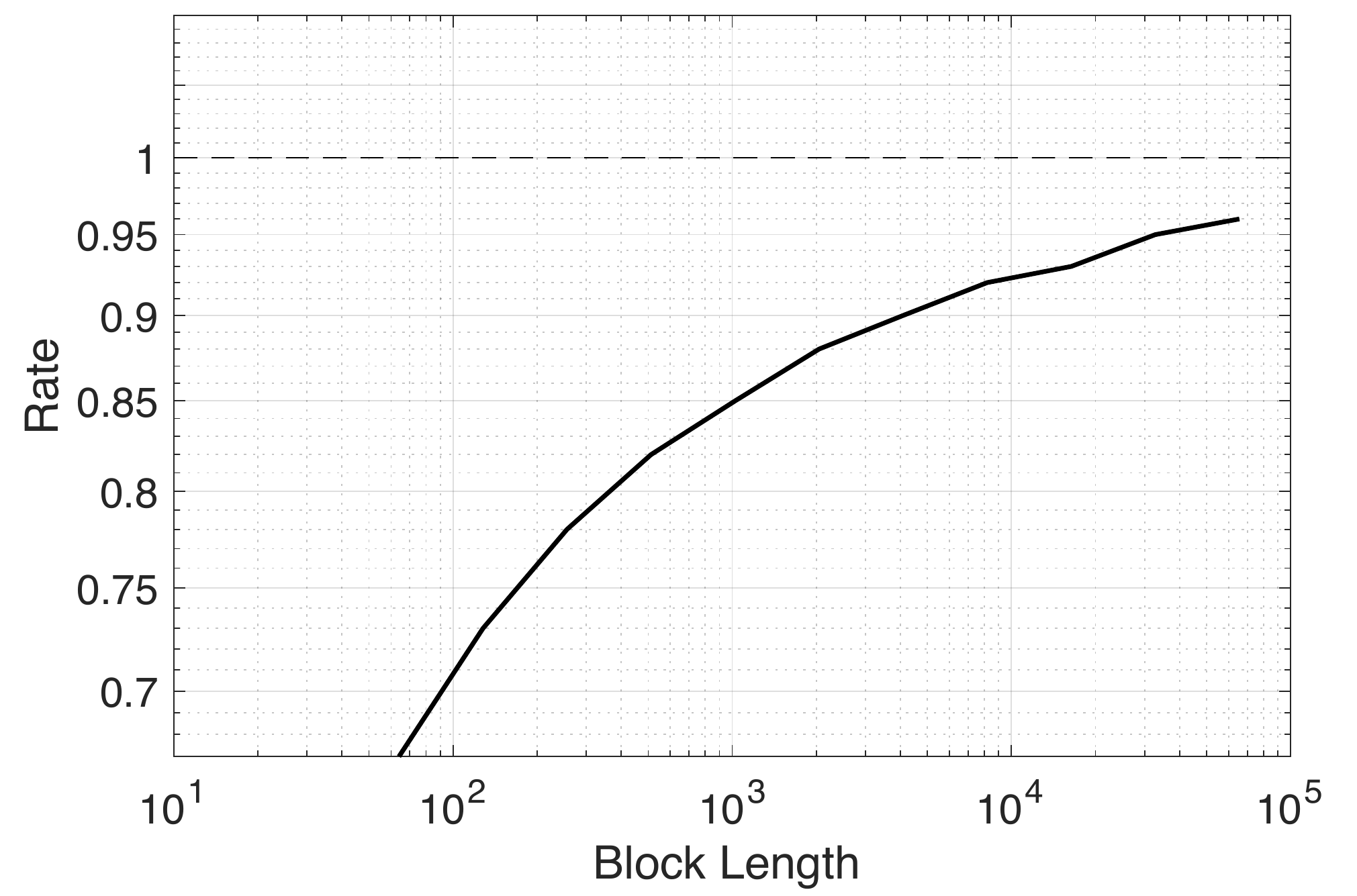}
\vspace{-0.1cm}
\caption{Achievable rates for the finite block length by strange attractor.} 
\label{fige1}
\end{figure} 

Rate2  in Table.\ref{tab1} is also achievable rates by strange attractor and some partial orders.
It can be noticed that the proposed partial ordering technique can increase the number of strange attractor set as seen in Example 5. and better bound for achievable rates can be found.

Now let's examine the asymptotic behaviour of the attractor set, which we are pre-defined. For this purpose, we provide the following expression. 

\begin{equation}
\lim_{N\rightarrow \infty}{\frac{\textsl{Number of channels with (11)}}{\textsl{Number of all channels}}}=1.
\end{equation}

Proof is given here. We consider the expression as follows.
The exact number of channels with $k_i\neq1$ and $k_{i+1}\neq1$ for any $i=[1,n)$ can be described as follows:
\begin{equation}
\Delta=\Delta_1+\Delta_2
\end{equation}
where $\Delta_1$ is shown in Fig.~\ref{fig_delta1} and $\Delta_2$ is shown in Fig.~\ref{fig_delta2}. 

\begin{figure}[t!]
\centering
\vspace{0.5cm}
\begin{tikzpicture}[thick,scale=0.5, every node/.style={scale=0.9}]
\draw (-0.2,-0.5)  -- (13.7,-0.5);
\draw (-0.2,0.5)  -- (13.7,0.5);
\draw (-0.2,-0.5)  -- (-0.2,0.5);
\draw (3.2,-0.5)  -- (3.2,0.5);
\draw (4.6,-0.5)  -- (4.6,0.5);
\draw (13.7,-0.5)  -- (13.7,0.5);
\draw (3.2,0-0.03) node[right] {$011$};
\draw (3.4,-1) node[right] {$3$};
\draw (6.6,0-0.03) node[right] {$all-possible$};
\draw (6.9,-1) node[right] {$n-t-3$};
\draw (1,0-0.03) node[right] {$\mathcal{A}^t$};
\draw (1,-1) node[right] {$t$};
\end{tikzpicture}
\caption{A graphical representation of the case $\Delta_1$.}
\label{fig_delta1}
\end{figure}

\begin{figure}[t!]
\centering
\vspace{0.5cm}
\begin{tikzpicture}[thick,scale=0.5, every node/.style={scale=0.9}]
\draw (-0.2,-0.5)  -- (13.7,-0.5);
\draw (-0.2,0.5)  -- (13.7,0.5);
\draw (-0.2,-0.5)  -- (-0.2,0.5);
\draw (0.9,-0.5)  -- (0.9,0.5);
\draw (13.7,-0.5)  -- (13.7,0.5);
\draw (-0.2,0-0.03) node[right] {$11$};
\draw (-0.2,-1) node[right] {$2$};
\draw (1.6,0-0.03) node[right] {$all-possible$};
\draw (1.9,-1) node[right] {$n-2$};
\end{tikzpicture}
\caption{A graphical representation of the case $\Delta_2$.}
\label{fig_delta2}
\end{figure}
Here,

\begin{equation}
\Delta_1=\sum_{t=0}^{n-3} |\mathcal{A}^t| \cdot 2^{n-t-3}
\end{equation}
and 
\begin{equation}
\Delta_2=2^{n-2}.
\end{equation}

We can show that
\begin{equation}
\Delta=2^{n-2}+\sum_{t=0}^{n-3} |\mathcal{A}^t| \cdot 2^{n-t-3}
\end{equation}

\begin{eqnarray}
\Delta&=&\left[\sum_{t=0}^{n-3} F_{t+2} \cdot 2^{n-t-3}\right]+2^{n-2}\nonumber \\
&=&2^{n-1}\left(\left[\sum_{t=0}^{n-3} F_{t+2}/2^{t+2}\right]+1/2\right).\nonumber 
\end{eqnarray}
Then, we have
\begin{equation}
\Delta=2^{n-1}\left(\left[\frac{F_2}{2^2}+\frac{F_3}{2^3}+\dots+\frac{F_{n-1}}{2^{n-1}}\right]+\frac{F_1}{2^1}+\frac{F_0}{2^0}\right).
\end{equation}

Here, notice that $\frac{F_0}{2^0}=0$ and $\frac{F_1}{2^1}=1/2$.
The final exact expression is
\begin{equation}
\Delta=2^{n-1}\sum_{t=0}^{n-1}F_t/2^t.
\end{equation}
There is the power series $\sum_{t=0}^{\infty} F_{t} \cdot k^{-t}=\frac{k}{k^2-k-1}$ for integer $k>1$.

As a result, 
\begin{equation}
\lim_{n\rightarrow \infty}{\frac{\textsl{Number of channels with (11)}}{\textsl{Number of all channels}}}=\lim_{n\rightarrow \infty} \frac{2^n\frac{1}{2}\frac{2}{2^2-2-1}}{2^n}=1.
\end{equation}

Now we can consider here how we can benefit from the definition of the attractor in code design for an example. 

For this purpose, we consider the natural channel with LLR parameter is greater than $\pi/2$.
In this case, the synthetic channel indices to be identified by the attractor will start from the most significant bit position, and the different length sequences will be determined which will reduce the LLR value of the natural channel to less than $\pi/2$ as an inter-level LLR value, (please see the {\it Example~\ref{example_pi}}).
For more clarity, we provide the following plain text for efficient design.
\vspace{0.5cm}
\begin{mydefinition}
An efficient code design:
\begin{enumerate}[i)]
\item Define an attractor set $\Omega$ as a subset of $\{1,2,\dots,N\}$ (indices without 11 for $n$-bit binary index)
\item {\bf for} i=1,...,n\newline \,\,\,\,\,\,\,\,Apply $i$-th partial order operator to update $\Omega$\newline {\bf end}
\item Apply simplified Gaussian approximation\\ to compute an ordering for some of the complement of the set $\Omega$
\item Then define new synthetic channels that are less reliable than the natural channel.
\end{enumerate}
\end{mydefinition}
\vspace{0.5cm}

\begin{myexample}
We consider $n=6$ in this example. There are $64$ synthetic channels placed in Table~\ref{tab2}. Here, black bold face binary expansions denote a channel identified by the attractor (i.e. they do not have 11). There are $F_{n+2}=21$ this type of channels for $n=6$ that are universally less reliable than the natural channel.

We apply the multiple partial order to find more channels that are worse than the natural channel $W$.
When we consider the first order operator to increase the number of bad channels, there are not any new bad channel by removing 1 in the attractor. The result is guaranteed that it is placed in the attractor. Then, we can apply second order operator to find more synthetic channels.
For example; $W_{(101000)}$ is a member of the attractor (i.e., 01$\rightarrow$10). By using 2nd order operator, we have the following result. 
\begin{equation}
W_{(011000)}\prec W_{(101000)}\prec W.
\end{equation}

Finally, we can apply third order operator.
For example; $W_{(011100)}$ is a member of the bad channels that are union set of attractor and 2nd order operator. By using 3rd order operator (i.e., 0110$\rightarrow$1001), we have the following result. 
\begin{equation}
W_{(011100)}\prec W_{(101010)}\prec W.
\end{equation}
The synthetic channels found by multiple partial order are denoted by blue bold face in Table.~\ref{tab2}.

As a result, we have found $38$ synthetic channels that are worse than the natural channels for $LLR<\pi/2$. On the other hand $27$ of them is still worse than the natural channel for $LLR<\pi$ (i.e., they do not have 1 in the first bit position).  
\end{myexample}

\begin{table}[hp]\large
\caption{Example for Attractor and Multiple Partial Order}
\begin{center}
\begin{tabular}{|c|c|c|c|}
\hline
                  \bf 000000 &                   \bf 001000 &                  \bf 010000 & \color{blue} \bf 011000 \cr \hline
                  \bf 000001 &                   \bf 001001 &                  \bf 010001 & \color{blue} \bf 011001 \cr \hline
                  \bf 000010 &                   \bf 001010 &                  \bf 010010 & \color{blue} \bf 011010 \cr \hline
\color{blue}\bf 000011 & \color{blue} \bf     001011 &\color{blue} \bf 010011 &\color{gray}   \bf   011011 \cr \hline
                   \bf 000100 & \color{blue} \bf 001100 &                 \bf 010100 & \color{red} \bf 011100 \cr \hline
                   \bf 000101 & \color{blue} \bf 001101 &                 \bf 010101 &\color{gray}   \bf   011101 \cr \hline
\color{blue} \bf 000110 & \color{red}   \bf 001110 & \color{red} \bf 010110 &\color{gray}    \bf  011110 \cr \hline
  \color{blue} \bf 000111 & \color{green}  \bf    001111 &\color{gray}   \bf   010111 &\color{gray} \bf     011111 \cr \hline \hline
                   \bf 100000 &                   \bf 101000 &\color{green} \bf      110000 &\color{gray}  \bf    111000 \cr \hline
                   \bf 100001 &                   \bf 101001 &\color{gray}  \bf    110001 &\color{gray}  \bf    111001 \cr \hline
                   \bf 100010 &                   \bf 101010 &\color{gray}   \bf   110010 &\color{gray}  \bf    111010 \cr \hline
\color{blue} \bf 100011 &\color{gray}    \bf  101011 &\color{gray}  \bf    110011 &\color{gray}   \bf   111011 \cr \hline
                   \bf 100100 &\color{gray}    \bf  101100 &\color{gray}  \bf    110100 &\color{gray}   \bf   111100 \cr \hline
                   \bf 100101 &\color{gray}    \bf  101101 &\color{gray}  \bf    110101 &\color{gray}   \bf   111101 \cr \hline
\color{blue} \bf 100110 &\color{gray}    \bf  101110 &\color{gray}  \bf    110110 &\color{gray}   \bf   111110 \cr \hline
\color{gray}    \bf  100111 &\color{gray}    \bf  101111 &\color{gray}   \bf   110111 &\color{gray}   \bf   111111 \cr \hline
\end{tabular}
\end{center}
\label{tab2}
\end{table}%
(Black: Strange Attractor. Blue: 2nd Partial Ordering. Red: 3rd Partial Ordering. Green: Computation. $W \preceq$ Gray chanels)
\newline
Here, we provide the following description of the strange attractor and partial order (PO) for the block length $N=64$.  
\begin{enumerate}[Step \text{A} i.]
\item {\color{blue} \bf 000011} $\prec$ {\color{black} \bf 000101} (Strange Attractor)
\item {\color{blue} \bf 000110} $\prec$ {\color{black} \bf 010010 }(Strange Attractor)
\item {\color{blue} \bf 000111} $\prec$  {\color{black} \bf 010101 }(Strange Attractor)
\item {\color{blue} \bf 100011} $\prec$ {\color{black} \bf 100101 }(Strange Attractor)
\item {\color{blue} \bf 100110} $\prec$  {\color{black} \bf 101010 }(Strange Attractor)
\item {\color{blue} \bf 001011} $\prec$  {\color{black} \bf 101010 }(Strange Attractor)
\item {\color{blue} \bf 001100} $\prec$  {\color{black} \bf 010100 }(Strange Attractor)
\item {\color{blue} \bf 001101} $\prec$  {\color{black} \bf 010101 }(Strange Attractor)
\item {\color{red} \bf 001110} $\prec$  {\color{blue} \bf 100110} (see step 5: 2nd PO) OR\\ {\color{red} \bf 001110} $\prec$  {\color{black} \bf 010101} (Strange Attractor 3th PO)
\item {\color{blue} \bf 010011} $\prec$  {\color{black} \bf 010101} (Strange Attractor)
\item {\color{red} \bf 010110} $\prec$  {\color{blue} \bf 010011} (see step 10: 2nd PO) OR\\ {\color{red} \bf 010110} $\prec$ {\color{black} \bf 100101} (Strange Attractor 3th PO)
\item {\color{blue} \bf 011000} $\prec$  {\color{black} \bf 101000} (Strange Attractor)
\item {\color{blue} \bf 011001} $\prec$  {\color{black} \bf 101001} (Strange Attractor)
\item {\color{blue} \bf 011010} $\prec$  {\color{black} \bf 101010} (Strange Attractor)
\item {\color{red} \bf 011100} $\prec$  {\color{blue} \bf 001110}  (see step 9: 2nd PO) OR\\ {\color{red} \bf 011100} $\prec$ {\color{black} \bf 101010} (Strange Attractor 3th PO)
\end{enumerate} 
It can be seen that the strange attractor and 2nd or 3th partial order steps can complete the all set that is shown in the Table~\ref{tab2}. By the help of 3th partial steps for the red bold face synthetic channels can be defined individually. And hence, all of the i-xv steps A can be compute, independently.  

Simply,\newline 
$\Omega$(Strange Attractor; 1st, 2nd, 3th PO) $\rightarrow$ fully idependently computable i-xv steps A \newline
OR \newline
$\Omega$(Strange Attractor; 1st, 2nd PO) $\rightarrow$ xv steps dependently computable in 3 steps: \{(5,10),(9,11) and all\}.

Then, we can compute for $\pi/2$ and $\pi$:
\begin{enumerate}[Step \text{B} i.]
\item {\color{green} \bf 110000} $\prec$ {\color{green} \bf 001111} $\prec$ $W$ (by compute)

\vspace{1cm}
for $001111$:
$$\pi/2 \rightarrow 0.6 \rightarrow 0.05 \rightarrow \times16=0.8 < \pi/2$$
$$\pi \rightarrow \pi/2 \rightarrow 1/2 \rightarrow \times16=8 > \pi$$
\vspace{1cm}

for $110000$:
$$\pi/2 \rightarrow \pi \rightarrow 2\pi \rightarrow 4.3 \rightarrow x \rightarrow 1 \rightarrow 0.25  < \pi/2$$
$$\pi \rightarrow 2\pi \rightarrow 4\pi \rightarrow 10 \rightarrow 7.6 \rightarrow 4.5 \rightarrow 2.7 < \pi$$
\vspace{1cm}
\item {\color{green} \bf 110000} $\prec$ $W$ $\prec$  {\color{gray} \bf 11XXXX}  (can be verified for XXXX: from 0001 to 1111 by the computation in the previous Step B i.)
(Following steps can be verified by the computation in the step B i.)
\item $W$ $\prec$ {\color{gray} \bf 010111}    
\item $W$ $\prec$ {\color{gray} \bf 011011}   
\item $W$ $\prec$ {\color{gray} \bf 011101}   
\item $W$ $\prec$ {\color{gray} \bf 011110}  
\item $W$ $\prec$ {\color{gray} \bf 011111}  
\item $W$ $\prec$ {\color{gray} \bf 100111}  
\item $W$ $\prec$ {\color{gray} \bf 101011}  
\item $W$ $\preceq$ {\color{gray} \bf 101100}  
\item $W$ $\prec$ {\color{gray} \bf 101101}  
\item $W$ $\prec$ {\color{gray} \bf 101110}   
\item $W$ $\prec$ {\color{gray} \bf 101111}  
\end{enumerate}

Here, Step B x is a sample for synthetic channels that is equivalent to the natural channel performance for $\pi/2$.

Finally,\newline 
$\Omega$(Strange Attractor; 1st, 2nd, 3th PO, computation step b i) $\rightarrow$  identifies all channels.

\section{Results}\label{SecResult}
In this section, we first investigate partial order 1001 and 0110 for some specific channels. Then, we provide performance results of Gaussian approximations that are compared to Monte Carlo simulations.
\pagebreak 
\subsection{Investigation for BEC}
We first consider BEC for partial order 1001 and 0110. As can be shown in the following figure Fig.~\ref{fige1}, analytic and simulation results show that 1001 is better than 0110 for BEC.
\begin{figure}[hp]
\centering
\vspace{-0.5cm}
\includegraphics[width=0.48\textwidth]{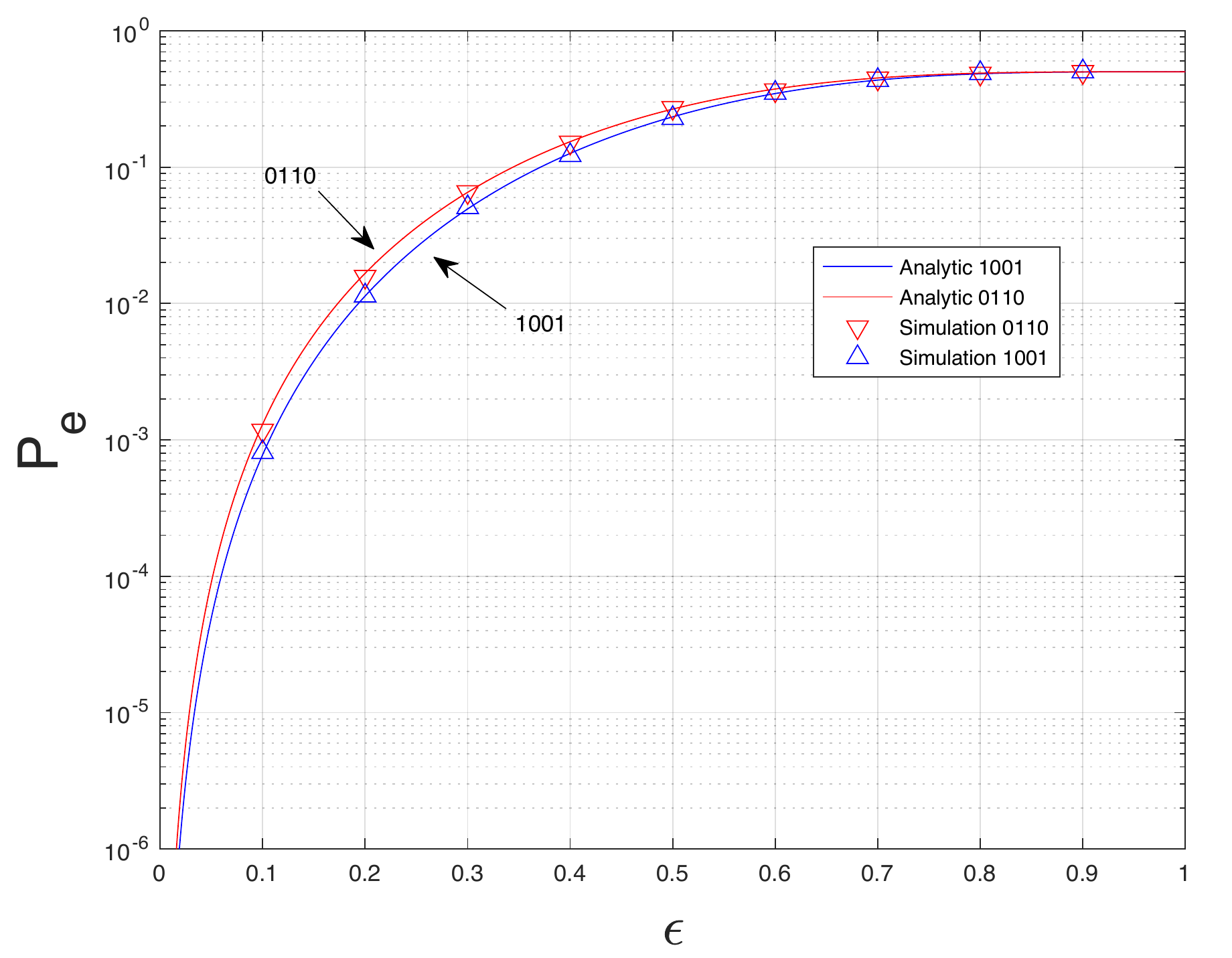}
\vspace{-0.1cm}
\caption{A numerical result of 1001 and 0110 for BEC. In this example 1001 is better than 0110.} 
\label{fige1}
\end{figure} 

\subsection{Investigation for AWGN Channel}

\begin{figure}[hp]
\centering
\includegraphics[width=0.5\textwidth]{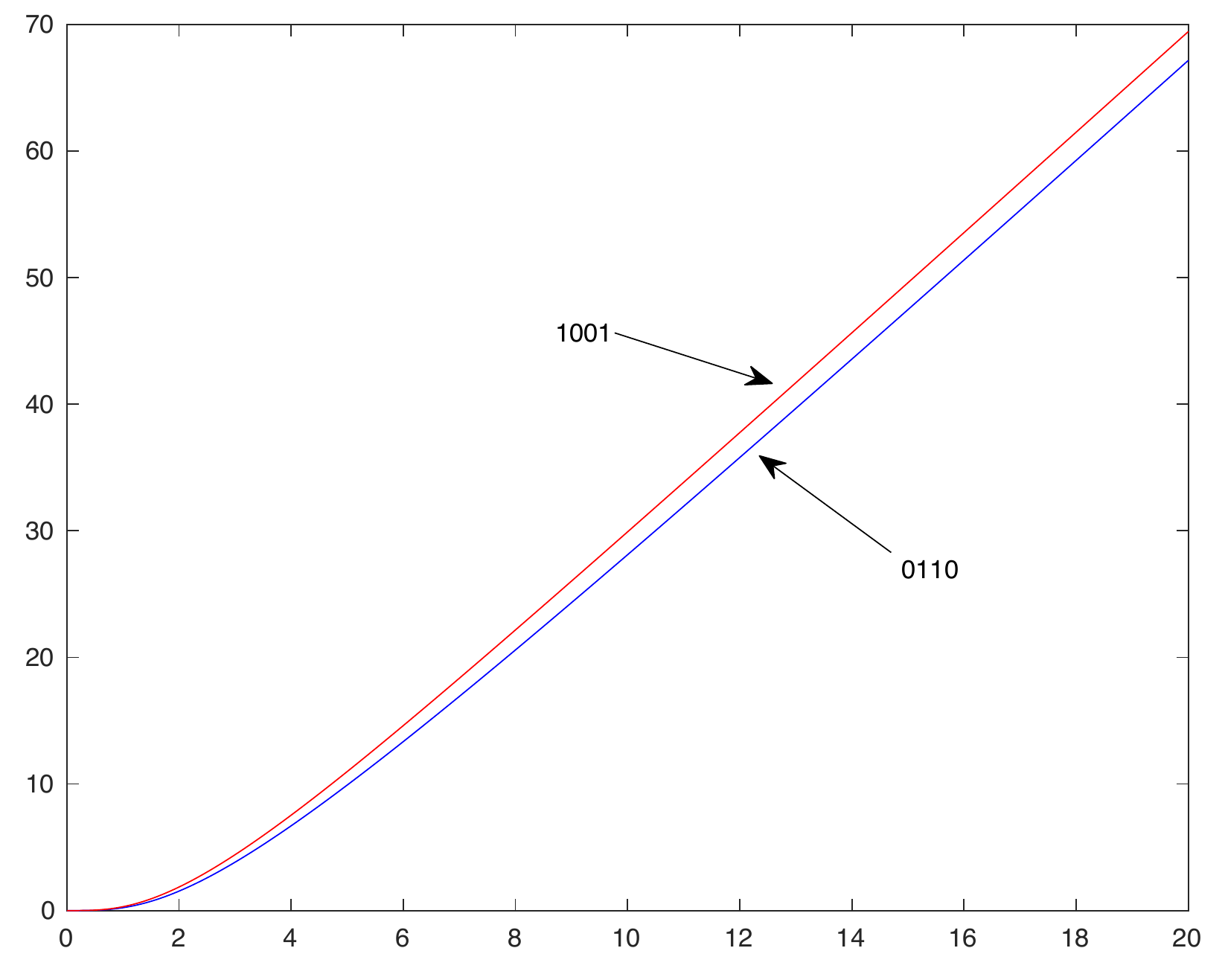}
\caption{By using a graphical result partial order 1001 is better than 0110 universally.} 
\label{fige3}
\end{figure} 
\pagebreak

Partial order 1001 and 0110 is investigated for AWGN channel by using Gaussian approximation method. As an example, following result show that 1001 is better than 0110 for a given fixed channel reliability. For this purpose we provide two figures as follows Fig.~\ref{fige2}. 
Moreover, we provide Fig.~\ref{fige3} the result for any given natural AWGN channel reliability that partial order 1001 is better than the 0110 for all channel conditions. 
\begin{figure}[b!]
\centering
\includegraphics[width=0.50\textwidth]{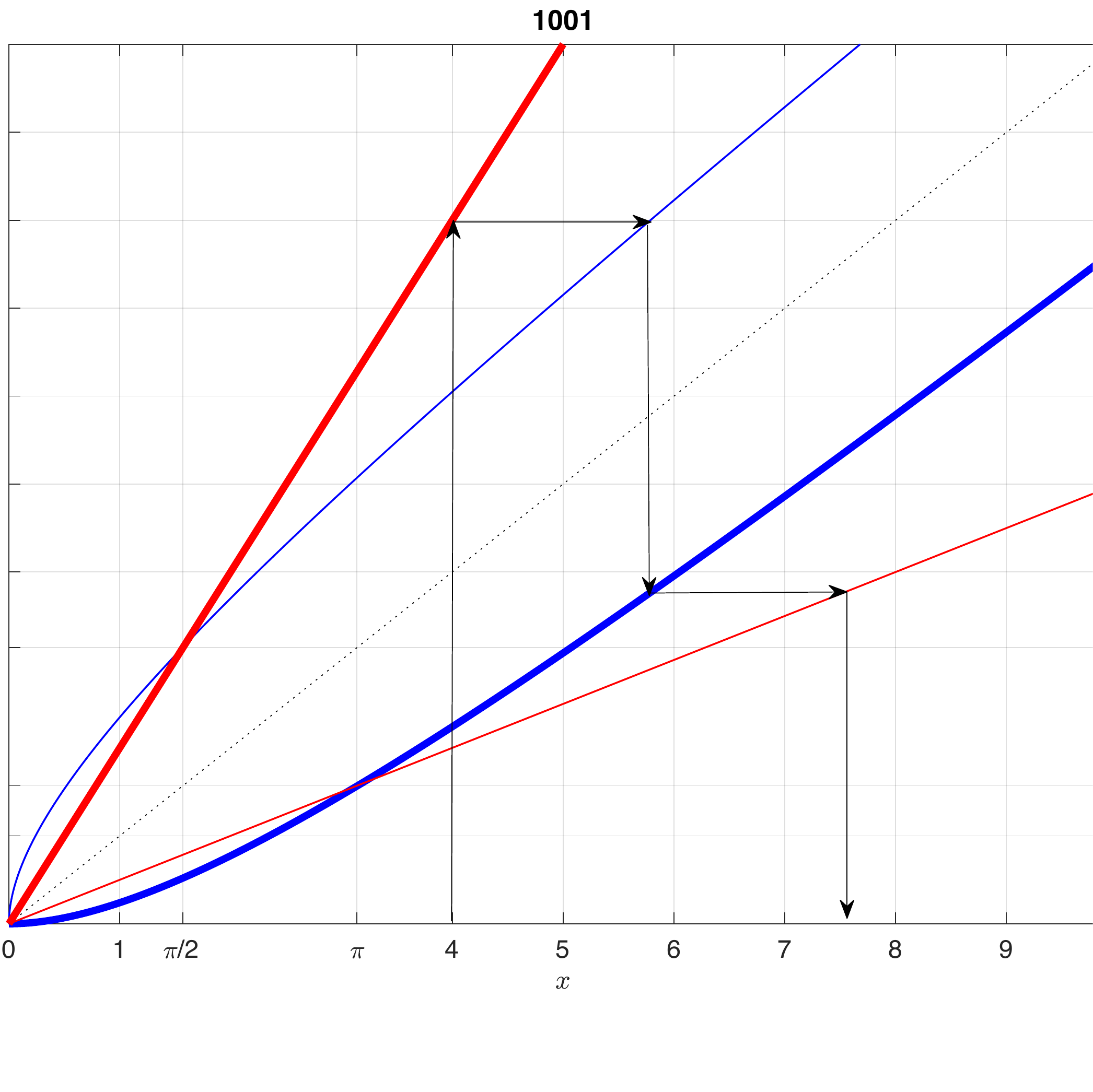}
\vspace{0.5cm}
\includegraphics[width=0.50\textwidth]{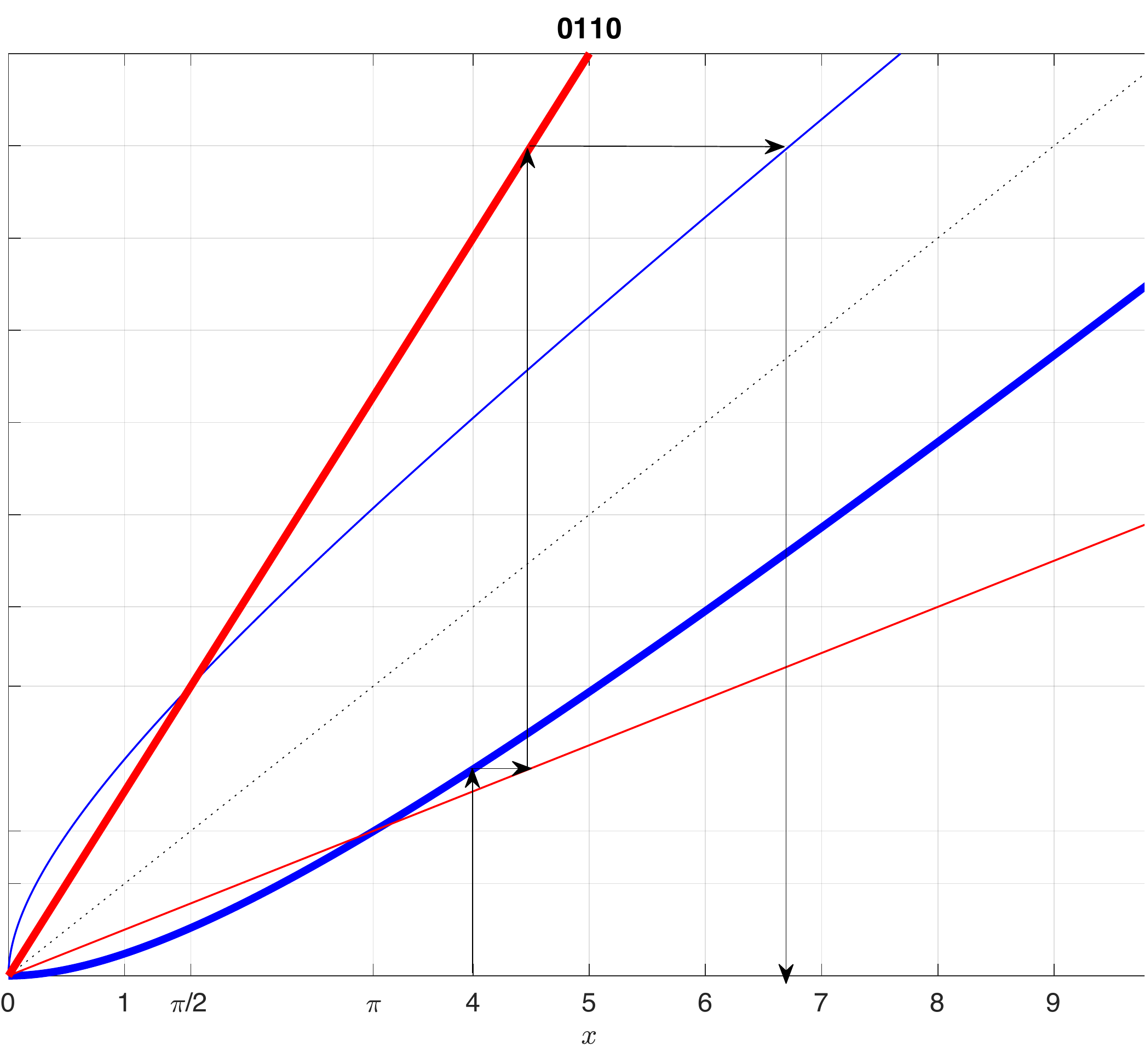}
\caption{A graphical representation of 1001 and 0110 for the same input llr. In this example 1001 (upper figure) is better than 0110 (lower figure).} 
\label{fige2}
\end{figure}

\pagebreak

\subsection{Investigation for BSC}
In this sub section, we investigate the error performance of the 1001 and 0110 channel ordering under BSC.

\begin{figure}[hp]
\centering
\includegraphics[width=0.5\textwidth]{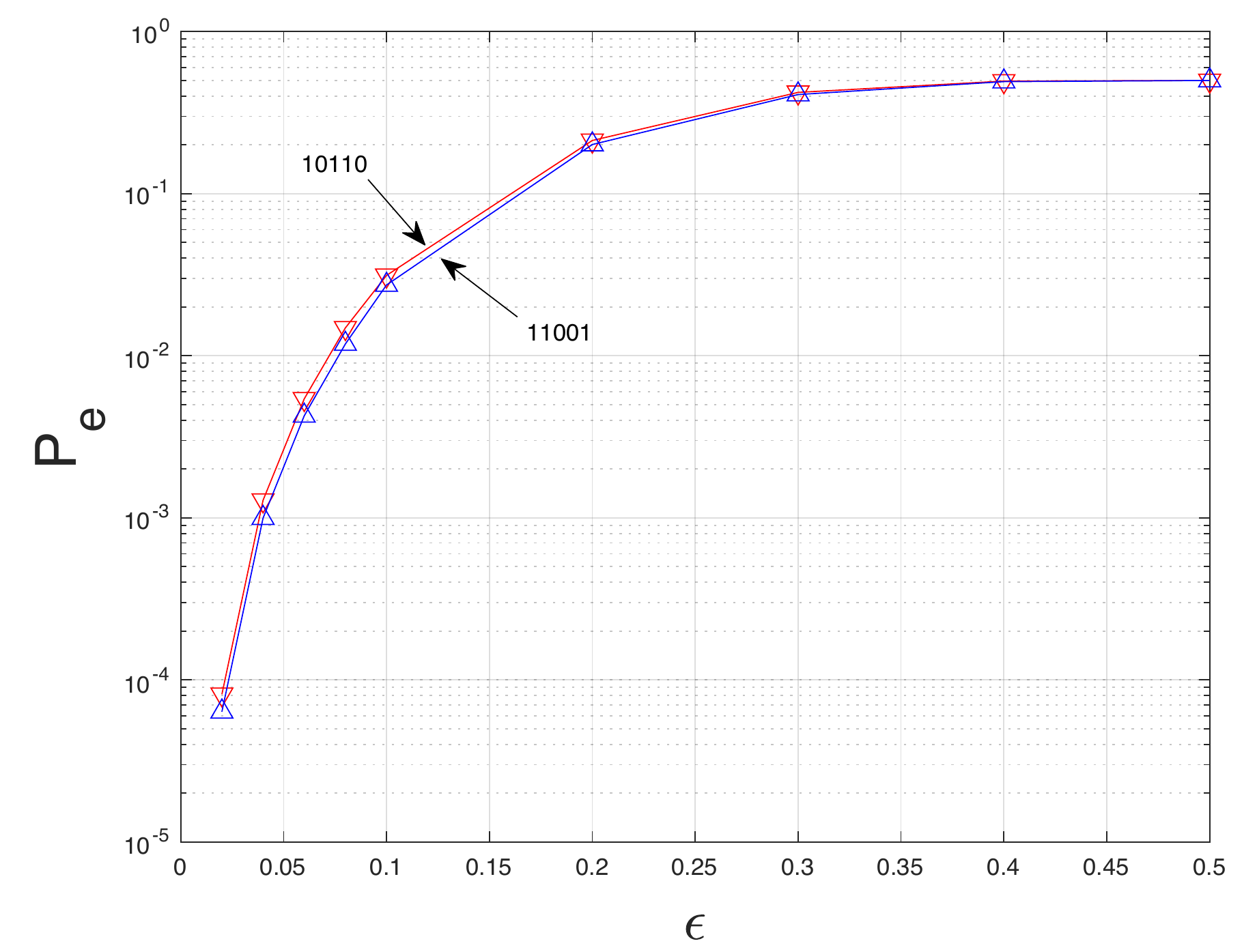}
\caption{A numerical result of $1{\color{black} \bf 1001}$ and $1{\color{black} \bf 0110}$ for BSC. In this example $1{\color{black} \bf 1001}$ is better than $1{\color{black} \bf 0110}$.} 
\label{fige4}
\end{figure} 

As a result of the investigations for BEC and BSC, it can be accepted that 1001 and 0110 is an universal property that result was also verified for AWGN channel.

\subsection{Performance comparisons of Gaussian Approximations}
Finally, Gaussian approximation methods were investigated in \cite{tufail} for high rate polar codes that are designed for optical communications. The simplified approximation proposed in this work are compared in the following figure. Simplified approximation is close to the Chung's method. They provide an upper bound for the Monte Carlo simulation. It is getting closer to the results of GA for high SNR region.    

\begin{figure}[hp]
\centering
\includegraphics[width=0.5\textwidth]{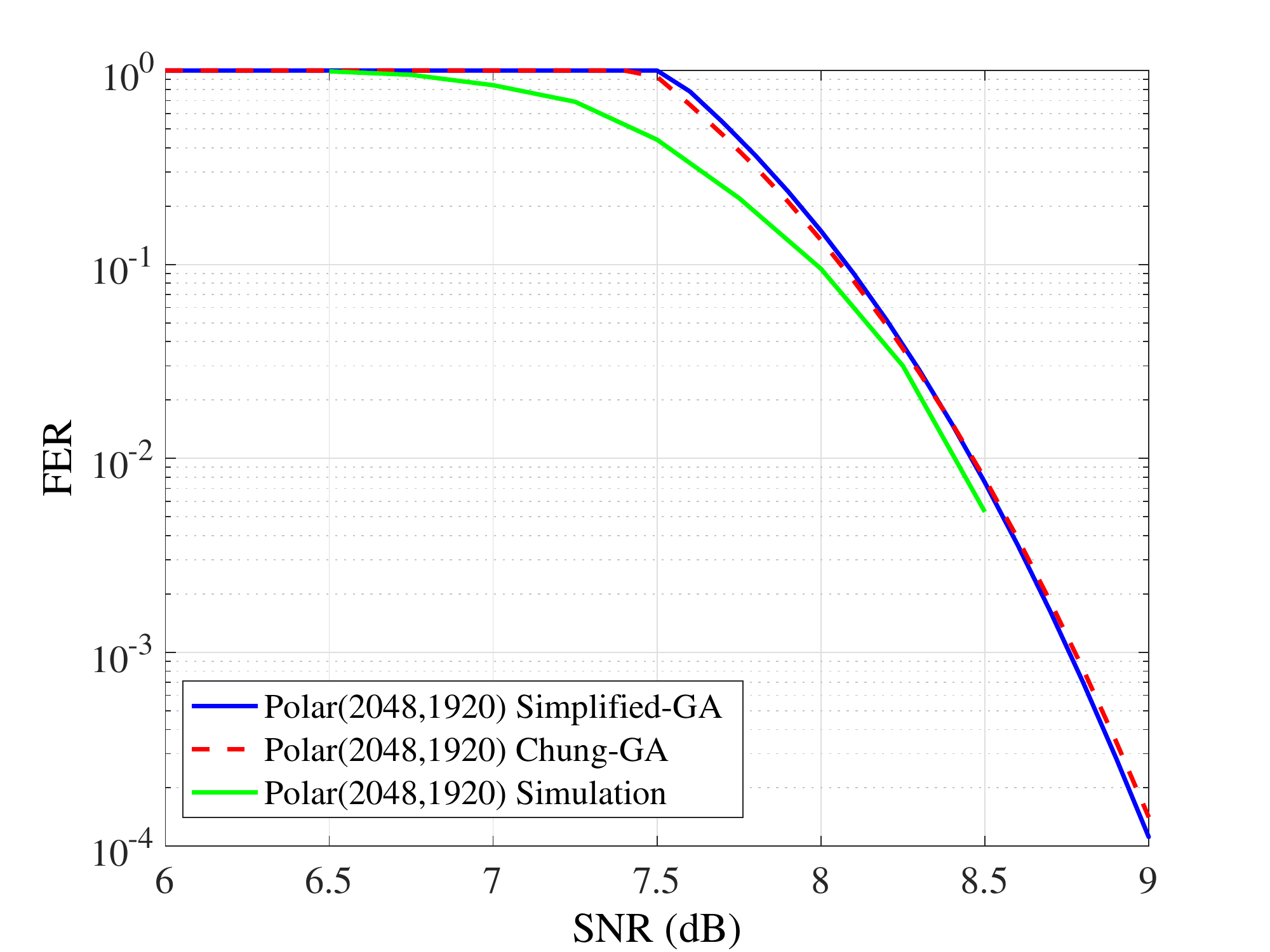}%
\caption{Simulation results and upper bounds by the Gaussian approximation (Chung's) method and the simplified-Gaussian approximation method.}
\label{fig_simpGA}
\end{figure}

\section{Conclusion}\label{sec7}
We focus on polar code constructions in this study. 
Universal (channel independent) properties of polar codes provide significant advantage for efficient constructions.
First, we introduced new partial ordering for polar codes as a natural result of left swap operator.
Then, the Gaussian approximation method was simplified by the help of simple recursive update rules that can be implemented by only a look up table.
As one of the main contributions, strange attractor was introduced not only for efficient code construction but also defining achievable rates in finite block lengths.
We show that cardinality of the strange attractor is related with Fibonacci numbers.   
Finally, we considered new partial orders and strange attractor for efficient code constructions and defining achievable rates at the finite block lengths.

%

%
\section*{Acknowledgment}
This work was performed in postdoc study at Bilkent University in Nov. 2015 - Nov. 2017 and the short visit at The Hong Kong Polytechnic University in Dec. 2018 - Jan. 2019. During the postdoc study author was supported by the Scientific and Technological Research Council of Turkey (T\"{U}B\.ITAK), grant: 1929B011500065. 
This work was partly presented at ISTC 2018 in Hong Kong with the same title.
Author would like to thank Prof. Erdal Ar{\i}kan (Bilkent University), Prof. Ruediger Urbanke (EPFL) and Prof. Francis C.M. Lau (The Hong Kong Polytechnic University) for helpful communications. 

\ifCLASSOPTIONcaptionsoff
  \newpage
\fi

\bibliography{attractor_Ref}{}

\begin{thebibliography}{10}
\providecommand{\url}[1]{#1}
\csname url@samestyle\endcsname
\providecommand{\newblock}{\relax}
\providecommand{\bibinfo}[2]{#2}
\providecommand{\BIBentrySTDinterwordspacing}{\spaceskip=0pt\relax}
\providecommand{\BIBentryALTinterwordstretchfactor}{4}
\providecommand{\BIBentryALTinterwordspacing}{\spaceskip=\fontdimen2\font plus
\BIBentryALTinterwordstretchfactor\fontdimen3\font minus
  \fontdimen4\font\relax}
\providecommand{\BIBforeignlanguage}[2]{{%
\expandafter\ifx\csname l@#1\endcsname\relax
\typeout{** WARNING: IEEEtran.bst: No hyphenation pattern has been}%
\typeout{** loaded for the language `#1'. Using the pattern for}%
\typeout{** the default language instead.}%
\else
\language=\csname l@#1\endcsname
\fi
#2}}
\providecommand{\BIBdecl}{\relax}
\BIBdecl

\bibitem{arikan_channel_2009}
E.~Ar\i{}kan, ``Channel polarization: A method for constructing
  capacity-achieving codes for symmetric binary-input memoryless channels,''
  \emph{IEEE Trans. Inf. Theory}, vol.~55, no.~7, pp. 3051--3073, Jul. 2009.

\bibitem{mori_density_eval}
R.~Mori and T.~Tanaka, ``Performance of polar codes with the construction using
  density evolution,'' \emph{IEEE Commun. Lett.}, vol.~13, no.~7, pp. 519--521,
  Jul. 2009.

\bibitem{tal_how_2013}
I.~Tal and A.~Vardy, ``How to construct polar codes,'' \emph{IEEE Trans. Inf.
  Theory}, vol.~59, no.~10, pp. 6562--6582, Oct. 2013.

\bibitem{trifonov_gaussian}
P.~Trifonov, ``Efficient design and decoding of polar codes,'' \emph{IEEE
  Trans. Commun.}, vol.~60, no.~11, pp. 3221--3227, Nov. 2012.

\bibitem{harish}
H.~Vangala, E.~Viterbo, and Y.~Hong, ``A comparative study of polar code
  constructions for the awgn channel,'' \emph{arXiv:1501.02473}, Jan. 2015.

\bibitem{schurch_partial}
C.~Sch\"{u}rch, ``A partial order for the synthesized channels of a polar
  code,'' in \emph{Proc. {{IEEE Int}}. {{Symp}}. {{Inform}}. {{Theory}}}, Jul.
  2016, pp. 220--224.

\bibitem{barded_partial}
M.~Bardet, V.~Dragoi, A.~Otmani, and J.-P. Tillich, ``Algebraic properties of
  polar codes from a new polynomial formalism,'' in \emph{Proc. {{IEEE Int}}.
  {{Symp}}. {{Inform}}. {{Theory}}}, Jul. 2016, pp. 230--234.

\bibitem{mondelli_sublinear}
M.~Mondelli, S.~H. Hassani, and R.~Urbanke, ``Construction of polar codes with
  sublinear complexity,'' in \emph{Proc. {{IEEE Int}}. {{Symp}}. {{Inform}}.
  {{Theory}}}, Jun. 2017, pp. 1853--1857.

\bibitem{mondelli_sublinear_arxiv}
------, ``Construction of polar codes with sublinear complexity,''
  \emph{arXiv:1612.05295v4}, Dec. 2016.

\bibitem{Chung}
S.-Y. Chung, T.~J. Richardson, and R.~L. Urbanke, ``Analysis of sum-product
  decoding of low-density parity-check codes using a gaussian approximation,''
  \emph{IEEE Trans. Inf. Theory}, vol.~47, no.~2, Feb. 2001.

\bibitem{reportedd}
L.~Frederic and G.~M. Maggio, ``Analysis of the iterative decoding of ldpc and
  product codes using the gaussian approximation,'' \emph{IEEE Trans. Inf.
  Theory}, vol.~49, no.~11, pp. 2993--3000, Nov. 2003.

\bibitem{vv}
C.~\"Ozg\"ur, ``Fraktal geometri ve kaos,'' \emph{Lecture Notes, Mechanical
  Eng. Faculty, \.I.T.\"U.}, 2008.

\bibitem{tufail}
A.~Tufail, ``Polar codes for optical communications,'' \emph{Master Thesis,
  Bilkent University.}, 2016.

\end{thebibliography}
\bibliographystyle{IEEEtran}

\end{document}